\documentclass[conference,a4paper]{IEEEtran}

\usepackage{cite}

\usepackage{graphicx}

\usepackage{amsmath}
\usepackage{amssymb}
\usepackage{amsthm}
\allowdisplaybreaks

\usepackage{array}
\usepackage{url}

\usepackage[disable]{todonotes}

\usepackage[final]{listings}
\usepackage{hyphenat}

\lstdefinelanguage
[x64]{Assembler}     % add a "x64" dialect of Assembler
[x86masm]{Assembler} % based on the "x86masm" dialect
{morekeywords={movq, shlq, orq, addq, retq, movl, rax, rip, rcx}} % etc.

\newcommand{\Abbildungps}[5]{%
	\begin{figure}[#1]%
		\begin{center}
			\scalebox{0.9}{\includegraphics*[width=#2\textwidth]{#3}}%
			\caption{#5}%
			\label{#4}%
		\end{center}
		\vspace{-0.43cm}
	\end{figure}%
}

\newcommand{\N}{\mathbb{N}}

\usepackage{algorithm}
\usepackage{algpseudocode}

\algtext*{EndWhile}% Remove "end while" text
\algtext*{EndIf}% Remove "end if" text
\algtext*{EndFunction}% Remove "end function" text
\algtext*{EndFor}% Remove "end for" text

\begin{document}
\title{Enabling Cross-Event Optimization in Discrete\hyp{}Event Simulation\\Through Compile-Time Event Batching}

\author{\IEEEauthorblockN{Marc Leinweber\\ and Hannes Hartenstein}
\IEEEauthorblockA{Karlsruhe Institute of Technology\\
Karlsruhe, Germany}
\and
\IEEEauthorblockN{Philipp Andelfinger}
\IEEEauthorblockA{Nanyang Technological University\\
Singapore}
}


\maketitle
%% Enforce page numbers in conference mode
\thispagestyle{plain}
\pagestyle{plain}

\begin{abstract}
A discrete-event simulation (DES) involves the execution of a sequence of event handlers dynamically scheduled at runtime. As a consequence, a priori knowledge of the control flow of the overall simulation program is limited.
In particular, powerful optimizations supported by modern compilers can only be applied on the scope of individual event handlers, which frequently involve only a few lines of code.
We propose a method that extends the scope for compiler optimizations in discrete-event simulations by generating batches of multiple events that are subjected to compiler optimizations as contiguous procedures.
A runtime mechanism executes suitable batches at negligible overhead.
Our method does not require any compiler extensions and introduces only minor additional effort during model development.
The feasibility and potential performance gains of the approach are illustrated on the example of an idealized proof-of-concept model. We believe that the applicability of the approach extends to general event-driven programs.
\end{abstract}

\section{Introduction}

In discrete-event simulations (DES), events are executed one after the other in the order of their time stamps. Due to the stochastic nature of most simulation models, the execution order is not known a priori. In particular, the execution order is not available during compilation. Modern compilers for languages such as C++ support powerful optimizations to reduce execution times and executable file sizes~\cite{grune-compiler} such as removal of redundant computations, instruction reordering to improve branch prediction and hide memory access latencies, or inline expansion of function calls to reduce call overheads. Unfortunately, the unpredictable control flow of discrete-event simulations limits the scope of such optimizations to individual event handlers. As an example, consider a situation where an event reverses the state changes performed by its preceding event. Although the execution of the two events has no effect, the situation cannot be detected by the compiler and this unnecessary computation is performed in full.

In this paper, we propose a general method to extend the scope for compiler optimizations in discrete-event simulations implemented in C++ beyond individual events. The approach can be applied both to sequential and parallel simulations. Although our implementation relies on C++ templates, the method only requires suitable metaprogramming facilities that are available in a number of programming languages.

\Abbildungps{t}{0.5}{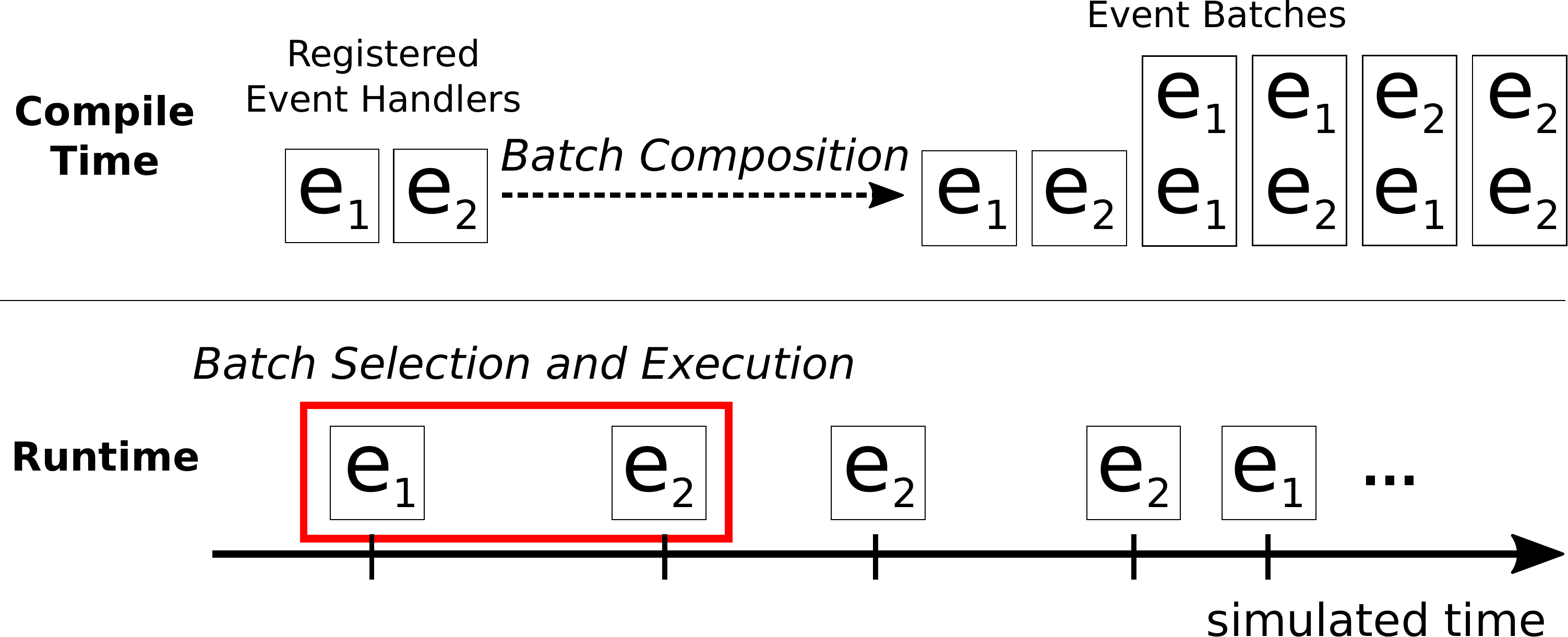}{img:sketch}{Overview of the proposed approach: during compilation, batches are composed from registered event handlers. At runtime, batches are identified in the set of pending events and executed as a unit.}

The approach is illustrated in Figure~\ref{img:sketch}. During compilation, all combinations of event handlers (``batches'') up to a configurable length are composed automatically using metaprogramming constructs. Each batch is a simple concatenation of the code associated with the individual event handlers. The benefit of the approach is that batches can be considered for compiler optimizations as contiguous code fragments, substantially improving the scope for optimizations. During runtime, appropriate batches are selected and executed in place of the original sequence of event handlers. The length of the executed batches varies, since similarly to parallel and distributed simulations~\cite{FujParallel}, dependencies between events must be respected. Model knowledge is applied to guarantee that batches are executable without violating the correctness of the simulation. We make our prototypical implementation of the approach available to the community~\footnote{\url{https://github.com/batched-DES/prototype}}.

The remainder of the paper is organized as follows: in Section~\ref{sec:related_work}, we discuss related work. In Section~\ref{sec:procedure}, we present a method to compose event batches at compile time and a mechanism to execute suitable batches at runtime. In Section~\ref{sec:evaluation}, we evaluate the approach w.r.t.~the effects on compilation and runtime performance, as well as the impact on model development. With Section~\ref{sec:conclusion} the results are summarized and the paper is concluded.

\section{Related Work}
\label{sec:related_work}

Some previous works have considered events as batches to reduce simulation runtimes. However, in contrast to our work, the existing approaches have focused on parallelized simulations and have not considered enabling compiler optimizations across events.

In optimistically synchronized parallel and distributed simulation, a rollback mechanism is required to undo erroneous event executions. Several algorithms have been proposed to reduce the scale and overhead of rollbacks~\cite{FujParallel}. Zeng et al.~proposed an approach to roll back batches of events at once. This is achieved by maintaining sufficient information at each processor so that during a rollback, a correct state can be achieved by rolling back a batch of events locally without inter-processor communication~\cite{batchcancel}.

% The authors illustrate a speedup in the range from 1.09 to 1.19.

In 2017, Gupta and Wilsey~\cite{wilseyGroup} performed experiments on executing batches of events in optimistically synchronized simulations using the framework~\textit{Warped2} to decrease the contention created by concurrent accesses to the shared data structure holding future events.
Comparing the performance under different policies for executing events as batches, the authors report a speedup of up to $2.5$ over executing events one after the other.

The idea of merging computations into a single step has been explored in general HPC (``superblock technique''~\cite{superblock}) and in the context of computations on graphics cards (``thread coarsening''~\cite{coarsening} and ``kernel fusion''~\cite{thrust}). In contrast to these works, our approach addresses the challenges given by the limited predictability of the order and dependencies among computations in DES.

%According to their observations in older publications and tests in their own parallel and distributed simulation framework \textit{Warped2}, it is possible to extract and execute more than one event at once. They discuss two approaches: Chain Scheduling and Block Scheduling. The chain scheduler takes the first $n \in \mathbb{N}$ elements of one LP and executes them, the block scheduler takes the first $n$ elements of all LPs placed on the computation node (in a distributed environment) and executes them. Speedup was measured in comparison to non-batched executions of \textit{Warped2}. Depending on the model and the scheduler in use, the authors report a speedup up to 2.5. However, they focus explicitly on optimizations applied during runtime.

\section{Proposed Method}
\label{sec:procedure}

In the following section we propose a method for batch composition and scheduling. The main steps of the approach are as follows:
\begin{itemize}
	\item \textbf{Enumeration and composition} of all combinations of event types up to a predefined length during compilation.
	%	\item Unique identification of a batch, based on the sequence of event types it is comprised of.
	\item \textbf{Selection and execution} of the correct batches during runtime depending on the set of pending events without violation of the causality constraint (non-decreasing time stamp order of the executed events).
\end{itemize}

%This problems are faced in the pipeline illustrated in Fig.~\ref{img:sketch}.
The batches of event handlers are composed during compilation. The input for this process is the set of event handling functions which are registered in an array by the modeler. In Section~\ref{sec:batch_composition}, a compile-time algorithm for the batch composition is described. We developed an event scheduler that maintains the execution order when executing batches instead of single events. Events are extracted as long as it can be guaranteed that the causality constraint is not violated. Subsequently, the corresponding batch is selected and executed. In Section~\ref{sec:batch_scheduling}, the scheduling method is presented.

The event batches are constructed during compilation through compile-time metaprogramming using C++ templates. Metaprograms are programs that manipulate executable code~\cite{metaprogramming}. Originally, C++ templates were intended as a mechanism for generic programming by allowing templated functions to be \textit{instantiated} for a specific data type during compilation~\cite{c++}. However, in 2003 Veldhuizen showed that C++ templates are Turing complete~\cite{templateTuring}.

\subsection{Batch Enumeration and Composition}
\label{sec:batch_composition}
As mentioned above, an event batch is a concatenation of $ n \in \N$ event handlers. To be able to execute the composed batches during runtime, the batches have to be identified uniquely. Hence, a system is needed that constructs identifiers based on the event types that contribute to a batch. In this Subsection, we present an algorithm that enumerates all possible batches up to a configurable length and that uses number system transformation based on a variation of the Horner scheme~\cite{horner}. We have chosen this approach to achieve a clean implementation and a scheme that can be efficiently evaluated both during runtime (cf.~Sec.~\ref{sec:batch_scheduling}) and compile-time.

The set of event handlers can be interpreted as the characters of an alphabet $\Sigma$. The resulting formal language of all event handler batches is $L = \Sigma^*$, where $*$ is the Kleene closure. For instance, the Kleene closure for the alphabet $\Sigma = \{a, b\}$ starts with~$\{\epsilon, a, b, aa, ab,...\}$, where $\epsilon$ is the empty string. Since formal languages are recursively enumerable \cite{formalLanguages}, a bijection $f$ can be found between $N$ and $L$. If the cardinality $|\Sigma|$ of $\Sigma$ is interpreted as the base of a number system, a word of~$\Sigma^*$ is a representation of a natural number in the system with base~$|\Sigma|$. The characters of the alphabet~$\Sigma$ are interpreted as digits of a number system. However, if the first character $a$ corresponds to the digit $0$, it has no effect on the batch identifiers ($aba$ will have the same id as $ba$). Hence, we introduce an explicit character $\nu$ signifying ``no event''. However, by including the $\nu$-event, redundant batches are generated. If $\Sigma = \{a\}$, it has to be expanded to $\Sigma_{\nu} = \{\nu, a\}$. If additionally $n = 2$, the words of $\Sigma_\nu^*$ with maximum length 2 are ${\nu, a, \nu \nu, \nu a, a \nu, aa}$. Obviously, the codes of the batches $a$, $\nu a$ and $a \nu$ are equivalent and thus redundant (cf.~Sec.~\ref{sec:evaluation}). Let $|\Sigma_\nu|$ be the number of event handlers (including the $\nu$-event) and $n$ the maximum batch length. Then~$B = \sum_{i=1}^{n}|\Sigma_\nu|^i$ event batches exist. 

%Let $\Sigma$ be the alphabet as defined above. Then $|\Sigma| = 2$, $a$ is the digit for $0$ and $b$ is the digit for $1$. The word (i.e., batch) $ba$ then can be seen as $1 \cdot 2^1 + 0 \cdot 2^0 = 2$. However, the word $aba$ has the same number as $ba$: $0 \cdot 2^2 + 1 \cdot 2^1 + 0 \cdot 2^0 = 2$. Hence, it is necessary to introduce an explicit digit $\nu$ signifying ``no event''. Then $aba$ becomes $1 \cdot 3^2 + 2 \cdot 3^1 + 1 \cdot 3^0 = 16$. 

%For each batch size ($1 \leq i \leq n, i,n \in \N$), there are $|\Sigma_\nu|^i$ words.

%We apply a variation of the Horner scheme~\cite{horner} to transform an event identifier in base 10  to an event sequence (e.g., $aa$).

%For the batch generation, a variant of the Horner scheme~\cite{horner} is used. The number in presentation of the initial system is divided by the new base, the remainder of step $i$ is the $i$th digit of the number in the new system. The quotient of step $i$ is divided by the base in step $i + 1$ again. The procedure is repeated until the quotient becomes $0$. Here it is used to transform the batch index to the event identifiers and sequence positions.

We assume that during model development, function pointers to all event handlers have been added to a constant array. Now, during compilation, all batch identifiers up to $B$ are enumerated by recursive evaluation of template functions. For each enumerated identifier a template function is invoked that transforms the identifier to the different event types and that batches the corresponding event handling functions in the correct execution order. Function pointers to the composed batch handlers are stored in an array to enable the runtime mechanism (cf.~Sec.~\ref{sec:batch_scheduling}) to address and execute batches at runtime.

App.~\ref{app:pseudocode} lists the pseudocode for the meta program described above.

%the batches are enumerated by evaluating recursive template functions

%When all batches are enumerated, the procedure stops. Its result is an array called \texttt{batchHandlers} of function pointers. The pointers point to the very first instantiated \texttt{generate} calls in line \ref{lst-line:handlerAssignment}. All other recursive calls are unrolled within the first call's function body.

%The batch composition is subject to limitations given by the compiler. In g++ and clang, a configurable maximum template instantiation depth is defined. Additionally, the number of batches grows exponentially. Thus, compile time may be growing at an enormous rate.

\subsection{Batch Selection and Execution}
\label{sec:batch_scheduling}

To maintain correct simulation results, a batch should only be executed if the contained sequence of events cannot be affected by its execution. Here, as in conservatively synchronized parallel and distributed simulations, we require model knowledge to define a \emph{lookahead}~\cite{FujParallel}, i.e., a minimum delta between an event's execution and creation time stamps.

%\todo[inline]{the minimum calculation below needs a double-check. Marc: is correct but notation is not easy to read (indexed e): $t_{max} = (\{t_e + l_e | e \in FEL\})$, ändern in Richtung Lesbarkeit (siehe Anpassung, FEL def fehlt, evtl. neue Menge einführen)}
At runtime, our batch scheduler uses a dynamic lookahead window. We assume that a lookahead value is associated with each event type. Events are extracted one by one while their execution time lies within the dynamic lookahead window. 

\Abbildungps{b}{0.5}{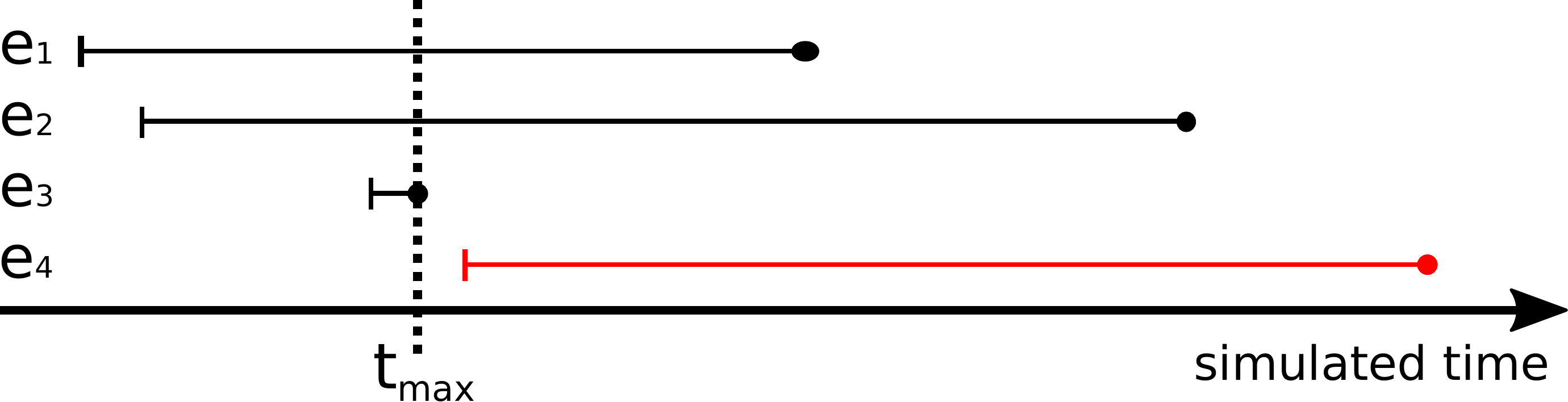}{img:lookahead}{Dynamic lookahead window (vertical line): according to the next events' time stamps and lookaheads, $t_{max}$ allows for the execution of a batch comprised of $e_1e_2e_3$.}

In~Fig.~\ref{img:lookahead}, the consideration of lookahead during batch extraction is illustrated. %First, event $e_1$ defines with its lookahead the maximum safe time stamp. Since $e_2$ has a smaller lookahead and we have advanced in simulation time, the maximum time stamp has to be reduced. The next event in queue, $e_3$, has a bigger lookahead window and the maximum time stamp is unaffected. Now, $e_4$ reduces the lookahead window again, which prohibits $e_5$ from being extracted and added to the batch.
We iterate over the future events and compute the minimum of the sum of the events' respective time stamp and lookahead. Once an event's time stamp is larger than the current minimum, the batch extraction terminates. In the figure, the batch $e_1 e_2 e_3$ is executed. In effect, if $t_e$ is the execution time of event $e$ and $l_e$ is the lookahead of $e$ according to its type, we compute $t_{max} = \min_{e \in E}(t_e + l_e)$ where $E$ is the set of future events up the configured maximum batch length. Following Section~\ref{sec:batch_composition}, the event type digits contribute, depending on their position in the batch, to the batch's identifier which then is used to execute the before composed batch. In addition to the lookahead window, the number of events in a batch is limited by the configured maximum batch length, which is a tuning parameter.

%The batch is identified, as described before, by a unique positive integer. The events' types contribute as a factor and their position in the sequence as an exponent to the cardinality of the alphabet (the amount of event handlers). The first event contributes with its event ID solely. The event ID is defined by the modeler and corresponds to the position of the event in the event handler array incremented by one as described in the section above.

%After the first event has been extracted, the batch is extended until the lookahead window is closed or the maximum batch size is reached. At first, it is checked if the next event in queue lies inbetween the lookahead window. If not, the current batch is closed and executed. Otherwise the loop continues. Thus, the code of up to $n$ events (the maximum batch size) is executed at once.

% Not all event types have the same lookahead. Thus, it is possible that the lookahead window has to be reduced. Hence, the lookahead of each event is considered and \texttt{maxTime} reduced, if needed. The event's position and ID are added to the batch's index. Then, the event is removed from the queue and the loop proceeds to the next iteration. 

\section{Evaluation}
\label{sec:evaluation}
In the following, we demonstrate the feasibility of the approach by showing successful cross-event compiler optimization and the associated speedup for a synthetic simulation model. Subsequently, we evaluate the increase in compile times incurred by the batching process. Finally, we discuss the deployability as well as limitations and potential improvements of our approach.

\subsection{Proof of Concept}
We evaluate our approach using a synthetic simulation model that performs redundant computations across events, providing substantial opportunities for cross-event compiler optimizations. The model is based on two event types: as a computationally intensive event, the \texttt{Increment} event performs a million iterations of \texttt{sum += sum + 1} on the global variable \texttt{sum}. The \texttt{Set} event sets the global \texttt{sum} variable to the constant value 10. For simplicity, neither of the event types schedules new events. However, our approach poses no limitations on event scheduling at runtime. Since the execution of a \texttt{Set} event after an \texttt{Increment} event renders the computation of the \texttt{for}-loop obsolete, the compiler should remove the loop from batches with this sequence of events entirely. 

A similar sequence of computations could be given in a model of wireless communication: suppose a node in a simulated network periodically broadcasts messages to nearby receivers. The successful reception depends on whether the receiver is in a power-saving state. If none of the nearby nodes is ready to receive, the computations involved in the creation of the message could be avoided entirely. If a sequence of events in a batch makes it impossible for the results of the message creation to be used, the compiler could remove the message creation code.

%\begin{lstlisting}[language={C++},caption=Computation-heavy C++ Increment Event as part of the proof-of-concept model,label={lst:incEvent},escapeinside={@}{@},float]
%static void run(Event* ev) {
%	for (size_t i = 0; i < 1'000'000; i++)
%		sum += sum + 1;
%}
%\end{lstlisting}
%
%\begin{lstlisting}[language={C++},caption=C++ Set Event as part of the proof-of-concept model that overwrites the results of the Increment Event,label={lst:setEvent},escapeinside={@}{@},float]
%static void run(Event* ev) {
%	sum = 10;
%}
%\end{lstlisting}

The model was compiled with \texttt{clang++} version 3.8.0 on an \texttt{Intel Core i5-6600K CPU~@~3.5GHz} with 32GB of main memory running Ubuntu 16.04.3. 

For an examination of the generated assembly code, we set the maximum batch length to 2 and thus composed $\frac{1 - (2+1)^3}{1 - (2+1)} - 1 = 12$ batches. The examination confirmed the successful cross-event optimization in this proof-of-concept example. As expected, when an \texttt{Increment} event is not followed by a \texttt{Set} event within the same batch, the compiler generates assembly code corresponding to the \texttt{Increment} event. However, if the \texttt{Increment} event is followed by a~\texttt{Set} event, the loop is omitted entirely, leaving only the assignment of the \texttt{Set} event. %The only instruction besides the return is to write an integer value of 10 into the global \texttt{sum} variable.

%\begin{lstlisting}[language={[x64]Assembler},caption=Assembly of batch ID 4: sequence of two Increment events),label={lst:batch4},escapeinside={@}{@},float]
%@\label{lst-line:forloop1}@mov     rax, qword ptr [rip + _ZN10POCModel3sumE]
%mov     ecx, 1000000
%.align  16, 0x90
%.LBB21_1:
%shl     rax, 8
%or      rax, 255
%add     rcx, -8
%@\label{lst-line:forloopExit1}@jne     .LBB21_1
%mov     qword ptr [rip + _ZN10POCModel3sumE], rax
%@\label{lst-line:forloop2}@mov     ecx, 1000000
%.align  16, 0x90
%.LBB21_3:
%shl     rax, 8
%or      rax, 255
%add     rcx, -8
%@\label{lst-line:forloopExit2}@jne     .LBB21_3
%mov     qword ptr [rip + _ZN10POCModel3sumE], rax
%ret
%\end{lstlisting}

%\begin{lstlisting}[language={[x64]Assembler},caption=Assembly of a Increment-Set-sequence: the loop is omitted,label={lst:batch7},escapeinside={@}{@},float]
%mov     qword ptr [rip + _ZN10POCModel3sumE], 10
%ret
%\end{lstlisting}

\subsection{Speedup}

We tested 24 different configurations, varying in the maximum batch length and the proportion $p_s$ of \texttt{Set} events. We set $p_s$ to $5\%$, $25\%$, $50\%$, and $75\%$. 

For each simulation run, $1\,000\,000$ initial events were scheduled. No new events were scheduled during the simulation. We scheduled one event at each integer time step. The lookahead was set to $1\,000\,000$ units of time, i.e., all executed batches had the maximum batch length.

\Abbildungps{t}{0.49}{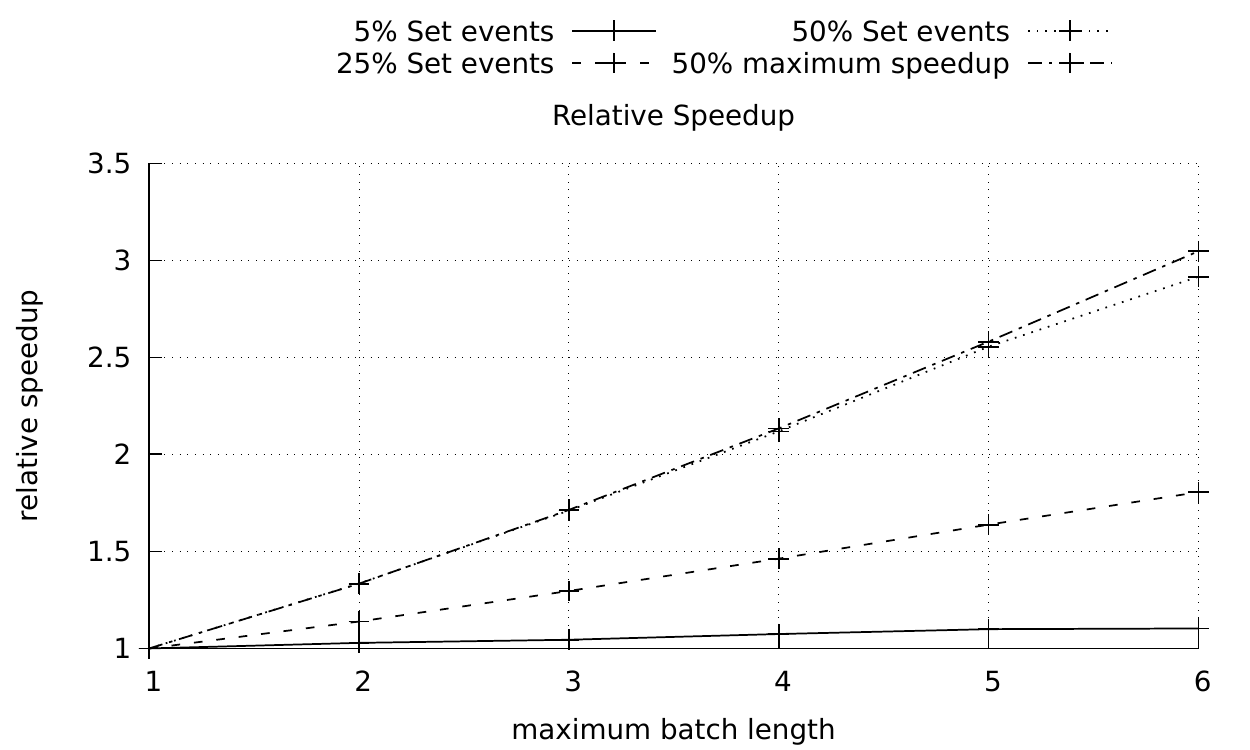}{img:speedup}{Speedup by event batching for synthetic proof-of-concept model.}

Each configuration was run $20$ times. For each run, a different seed for the pseudo-random number generator was used. In Fig.~\ref{img:speedup}, the result is plotted. The achieved speedup depends on the chance that the computation-intensive \texttt{Increment} event may be omitted.

We can observe that in an idealized case, event batching achieves substantial speedup.
Assuming that the impact of \texttt{Set} events on the runtime is negligible, given a fixed batch length $n$, the proportion $p_i = 1 - p_s$ of \texttt{Increment} events, the maximum possible speedup $s_\text{max}$ can be calculated based on the expected number of \texttt{Increment} events in each batch, and the probability that an \texttt{Increment} event at a certain position in the batch will not be followed by any \texttt{Set} events:
$s_\text{max} = {n{p_i}}({\frac{1 - {p_i}^{n+1}}{1 - p_i} - 1})^{-1} = \frac{n(1-p_i)}{1 - {p_i}^n}$ (see App.~\ref{app:proof} for a proof). In Figure~\ref{img:speedup}, $s_\text{max}$ is plotted for $p_s = 0.5$. With increasing batch lengths, $s_\text{max}$ approaches $n/2$. Our performance measurements closely approximate the maximum possible speedup.

%However, one would expect a maximum speedup of factor 2 when 50\% of the loops on average are omitted. In contrast, our benchmark shows a linear increase up to a factor of 3. The reason is that a batch of length consists in average of 3 \texttt{Increment} and 3 \texttt{Set} events. All of those \texttt{Increment} events are omitted if the batch ends with a \texttt{Set} event.

To investigate the overhead incurred purely by the runtime batch selection, we compared the runtime with a one-by-one event execution as in common sequential discrete-event simulators. When executing $m$ \texttt{Set} events in the unbatched case and with our batching approach, we observed that the runtime batch selection adds overhead of about 5\% at an average batch length of 2.

\subsection{Compilation}

To benchmark the compile-time algorithm (cf.~Sec.~\ref{sec:batch_composition}), four configurations, varying in the number of different event types, were compiled 15 times. Once the compilation time for a single configuration exceeded 240 seconds, the compilation was stopped. The compile times are plotted in Fig.~\ref{img:compiletimes}. As expected from the increase in the batch count, the compile times increase exponentially. With ten different event types and a maximum batch length of 5, the compilation time exceeds 240 seconds drastically. 

\Abbildungps{t}{0.49}{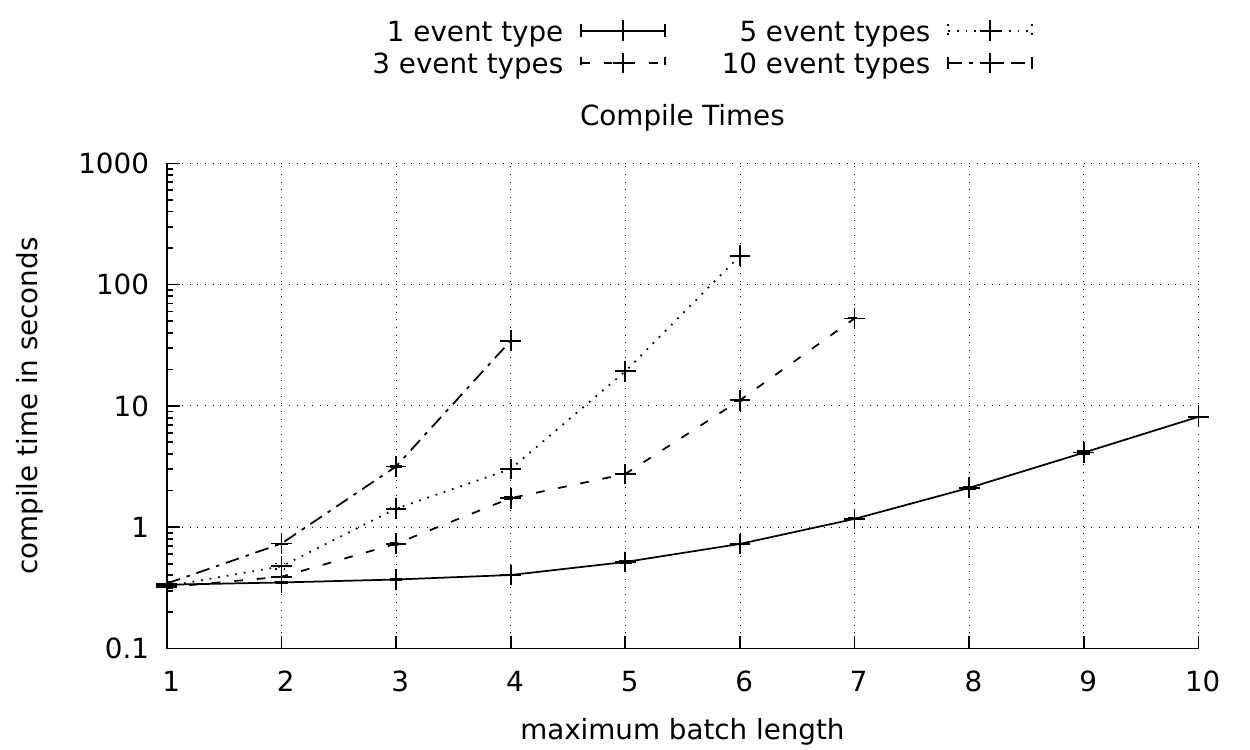}{img:compiletimes}{Effect of event batching on compile times.}

Although these numbers seem enormous, large C++ projects like simulation models can easily reach compile times in ranges of several minutes. As long as the number of event types is small, the relative impact on compilation times may be acceptable. Further, the compilation times can be reduced by choosing a smaller maximum batch length. Additionally, as discussed in Section~\ref{sec:batch_composition}, due to the need for a $\nu$-event, a substantial number of batches will never be used by the scheduler. Overall, $(\frac{1-(|\Sigma|+1)^{n+1}}{1-(|\Sigma|+1)} -1) - (\frac{1-|\Sigma|^{n+1}}{1-|\Sigma|} -1)$ redundant batches are composed. As an example, with five different event types and a maximum batch length of five, $9331$ batches (i.e., $58\%$) are redundant. In the future, a refined enumeration scheme could eliminate these redundant batches.

\subsection{Discussion}

In this section, we first state a number of limitations and technical aspects of the batching approach. Subsequently, we discuss potential avenues for future work.

Overall, the approach imposes only minor restrictions on the model development. A marginal burden on modelers is given by the assumption made in the batch composition that function pointers to all event handlers are initially stored in an array.

A key avenue for future work lies in reducing the overhead during compilation and at runtime.
Major opportunities for improvements are given within the batch composition process: since a large proportion of the composed batches are redundant, an improved enumeration scheme could reduce the number of batches.

During the execution of a batch, each event may generate new events immediately. By postponing the scheduling of all new events to the end of a batch execution, it may be possible to improve performance through a reduction in accesses to the data structure holding the scheduled events.

Currently, the runtime mechanism executes batches conservatively, i.e., never violating the simulation correctness. If the batch lengths achievable in this manner are small, a speculative approach could improve performance. Similarly to optimistically synchronized parallel simulations~\cite{FujParallel}, if the execution of a batch generates new events with time stamps earlier than the last event in the batch, a rollback mechanism could restore a correct simulation state.

Whether significant opportunities for cross-event compiler optimizations exist in real-world models is still to be investigated. For instance, in real-world simulation models, the assignment of events to simulated entities is typically only known at runtime. Thus, the compiler may not be able to determine whether a variable change by an event will be overwritten by a successor event.
It may be possible to formulate implementation guidelines for simulation models to maximize opportunities for optimizations.

A wide range of modern applications is implemented in an event-driven manner, i.e., events representing computational tasks are dynamically scheduled and extracted according to their priorities. We believe that the proposed batching approach is generic enough to be applied to event-driven applications beyond simulations.

Finally, since our approach relies on C++ template metaprogramming, the feasibility of the approach beyond C++-based simulators should be explored. Generally, the approach requires compile-time evaluated  metaprogramming facilities that support code transformation, code generation and reflection. However, reflection is not needed if constructs such as function pointers exist. Beyond ahead-of-time compilation, modern just-in-time compilers achieve remarkable optimization results through runtime profiling. Hence, just-in-time compilation could focus on the creation of relevant batches according to the observed runtime behavior of the considered simulation model.

\section{Conclusion}
\label{sec:conclusion}
We presented an approach that extends the scope of compiler optimizations in discrete-event simulation beyond individual events. Using C++ template metaprogramming, all possible sequences of events up to a configurable sequence length are composed at compile time. A runtime mechanism selects and executes event batches, relying on model-specific temporal properties to maintain correctness. We showed that the approach is feasible, achieves substantial speedup in a proof-of-concept example and does not add immense runtime overhead. Only minor additional effort is required by modelers to apply the approach. The main limitation is a substantial increase in compilation times and executable sizes, both of which may be improved on in future work by an advanced batch composition approach.
Further, we consider experiments with real-world discrete-event simulation models and general event-driven applications the most interesting avenues for further exploration.



%\bibliographystyle{IEEEtran}

\bibliographystyle{IEEEtran}
\bibliography{sample-bibliography}

\appendices

\newpage
\section{Batch Composition Meta Program Pseudo Code}
\label{app:pseudocode}

In Algorithm \ref{alg:batch_composition} the meta program pseudo code for the batch composition procedure, which takes place during compilation, is given. The function~\textsc{enumerateBatches()}~enumerates all possible batch IDs in a recursive manner. For each ID \textsc{genBatch()} is invoked where the different event types are identified with the modified Horner's scheme. The event types' handling functions are concatenated with \textsc{appendFuncCall()} into a single function body and the pointer to the resulting function is stored.

\begin{algorithm}
	
	\caption{Batch Composition Metaprogram}
	\label{alg:batch_composition}
	
	\footnotesize
	
	\begin{algorithmic}
		\\
		\State $eventHandlers$ : Array of $Functionpointer$
		\State $batchedHandlers$ : Array of $Functionpointer$
		\State $maxBatchSize$ : $\N$
		\State $eventTypeCount$ = $sizeof(eventHandlers)$ : $\N$
		\State $batchCount = \frac{1 - eventTypeCount^{maxBatchSize + 1}}{1 - eventTypeCount} - 1$ : $\N$
		\\ 
		\Function{generateBatches}{}
		\State $enumerateBatches(0)$
		\EndFunction
		\\
		\Function{enumerateBatches}{$batchID$ : $\N$}
		\If{$batchID = batchCount$} \Return
		\EndIf 
		\State $pointer = $ pointer to new empty function : $Functionpointer$
		\State $batchedHandlers[batchID] := genBatch(batchID, pointer)$
		\State $enumerateBatches(batchID + 1)$ \Comment{enumerate all indexes}
		\EndFunction
		\\
		\Function{genBatch}{$batchID$ : $\N$, $pointer$ : $Functionpointer$}
		\If{$batchID = 0$} \Return $pointer$ \Comment{batch is complete}
		\EndIf
		\State $eventIndex$ := $batchID \mathbf{\mod} eventTypeCount$ \Comment{factor of $k$th exponent}
		\If{$eventIndex > 0$} \Comment{check for $\nu$-event}
		\State $appendFuncCall(pointer, eventHandlers[eventIndex - 1])$
		\EndIf \\
		\State $genBatch((batchID / eventTypeCount), pointer)$ \Comment{continue with quotient}
		\EndFunction
		\\
	\end{algorithmic}    
\end{algorithm}

\newpage
\section{Proof of $s_{max}$}
\label{app:proof}

\newtheorem{lemma}{Lemma}
\newtheorem{satz}{Corollary}
\newcommand{\p}{{p_I}}
\noindent
$p_I\ \hat{=} $ probability of an \textit{Increment} Event\\
$p_S = 1 - \p \hat{=} $ probability of a \textit{Set} Event\\
$n\ \hat{=} $ batch length \\
$T_1 \hat{=} $ standard DES time \\
$T_p \hat{=} $ batched DES time \\
$S = \frac{T_1}{T_p} \hat{=} $ speedup\\
$E[S] = \frac{E[T_1]}{E[T_p]}$ \\
$E[T_1] = n \cdot \p$ \\
$E[T_p] = \sum_{j=1}^{n-1}(j \cdot \p^j \cdot p_S) + n \cdot \p$

\begin{lemma}{Closed form of $E[T_p]$.}
	\label{lemma}
	\begin{align*}
	E[T_p] = \sum_{j=1}^{n-1}(j \cdot \p^j \cdot p_S) + n \cdot p_i = \frac{1-\p^n}{\frac{1}{\p}-1}
	\end{align*}
\end{lemma}
\begin{proof}
	\begin{align*}
	E[T_p] &= \sum_{j=1}^{n-1} (j \cdot \p^j \cdot p_S ) + n \cdot \p^n
	\\ &= \sum_{j=1}^{n-1} (j \cdot \p^j \cdot (1-\p)) + n \cdot \p^n
	\\ &= (1-\p) \sum_{j=1}^{n-1} (j \cdot \p^j) + n \cdot \p^n
	\\ &= \sum_{j=1}^{n-1} j \cdot \p^j - \p \sum_{j=1}^{n-1} j \cdot \p^j + n \cdot \p^n
	\\ &= \sum_{j=1}^{n-1} j \cdot \p^j + n \cdot \p^n - \p \sum_{j=1}^{n-1} j \cdot \p^j
	\\ &= \sum_{j=1}^{n} j \cdot \p^j -\p \sum_{j=1}^{n-1} j \cdot \p^{j}
	\\ &= \sum_{j=1}^{n} j \cdot \p^j - \sum_{j=1}^{n-1} j \cdot \p^{j+1}
	\\ &= \sum_{j=1}^{n} j \cdot \p^j - \sum_{j=2}^{n} (j-1) \cdot \p^j
	\\ &= \sum_{j=1}^{n} j \cdot \p^j - \sum_{j=1}^{n} (j-1) \cdot \p^j
	\\ &= \sum_{j=1}^{n} (j - j + 1) \p^j
	\\ &= \sum_{j=1}^{n} \p^j
	\\ &= \p \sum_{j=1}^{n} \p^{j-1}
	\\ &= \frac{(1-\p) \cdot (\sum_{j=0}^{n-1}\p^j) \cdot \p}{1-\p}
	\\ &= \frac{(1-\p^n) \cdot \p}{1-\p}
	\\ &= \frac{1-\p^n}{\frac{1}{\p}-1} \qedhere
	\end{align*}
\end{proof}

\begin{satz}{Maximum speedup $s_{max}$.}
	\label{satz}
	The expected value for the maximum speedup is
	\begin{align*}
	s_{max} = E[S] = \frac{n \cdot (1 - \p)}{1 - \p^n}.
	\end{align*}
\end{satz}
\begin{proof}
	\begin{align*}
	E[S] &=  \frac{E[T_1]}{E[T_p]}
	\\ &= \frac{n \cdot \p}{\frac{1-\p^n}{\frac{1}{\p}-1}}
	\\ &= (n \cdot \p) \cdot \frac{\frac{1}{\p}-1}{1-\p^n}
	\\ &= \frac{n \cdot \p \cdot \frac{1}{\p} - n \cdot \p}{1 - \p^n}
	\\ &= \frac{\frac{n \cdot \p}{\p} - n \cdot \p}{1 - \p^n}
	\\ &= \frac{n - n \cdot \p}{1 - \p^n}
	\\ &= \frac{n \cdot (1 - \p)}{1 - \p^n} \qedhere
	\end{align*}
\end{proof}

\end{document}